\documentclass[10pt,conference]{IEEEtran}

\usepackage{amsmath,amssymb,amsthm}
\usepackage{graphicx}
\usepackage{booktabs}
\usepackage{xcolor}
\usepackage[colorlinks=true,allcolors=blue]{hyperref}
\usepackage{enumitem}

\newtheorem{proposition}{Proposition}
\newtheorem{corollary}[proposition]{Corollary}
\newtheorem{definition}{Definition}
\newtheorem{remark}{Remark}

\begin{document}

\title{Entanglement geometry separates circuit cutting, classical hardness, and trainability} 

\author{
\IEEEauthorblockN{María Gragera Garcés}
\IEEEauthorblockA{
\textit{Quantum Software Lab}\\
\textit{University of Edinburgh, UK}\\
 m.gragera.garces@ed.ac.uk}
\and
\IEEEauthorblockN{Sabina Dr\u{a}goi}
\IEEEauthorblockA{
\textit{ETH Z\"{u}rich}, Z\"{u}rich, Switzerland\\
\textit{IBM Research}, R\"{u}schlikon, Switzerland\\
sdragoi@phys.ethz.ch}
\and
\IEEEauthorblockN{Lirand\"{e} Pira}
\IEEEauthorblockA{
\textit{Centre for Quantum Technologies,}\\
\textit{National University of Singapore, Singapore}\\
lpira@nus.edu.sg}
}

\maketitle

\begin{abstract}
Circuit cutting promises to scale quantum computations beyond current hardware, but variational quantum advantage also requires low cutting overhead, classical hardness, and trainability.
We show that these properties are strongly constrained by entanglement geometry. 
Matrix product state (MPS) and tree tensor network (TTN) circuits with constant seam bond dimension can be cut with \(O(1/\varepsilon^2)\) sampling overhead, but remain efficiently classically simulable, ruling out asymptotic quantum advantage within these families. 
By independently controlling seam and intra-block entanglement, we construct a two-block circuit family that remains cheaply cuttable while requiring a super-polynomial global MPS bond dimension, as supported numerically up to \(n=100\). 
However, MPS hardness and trainability require incompatible depth regimes, \(d=\omega(\log n)\) and \(d=O(\log n)\), respectively. 
Using magic rather than entanglement as the hardness resource avoids this conflict: shallow Clifford+\(T\) circuits remain cuttable and trainable while their stabiliser-simulation cost grows exponentially with the \(T\)-count.
\end{abstract}


\section{Introduction}

Circuit cutting~\cite{peng2020,piveteau2024,lowe2023,piveteau2025} distributes a wide quantum circuit across small devices by replacing cross-device gates with local operations and classical post-processing.
Quantum links would enable lower overhead~\cite{piveteau2024,bechtold2023}, but are not yet broadly available.
Assuming the devices are connected by classical channels only, the sampling overhead grows exponentially in the number of cuts but stays independent of circuit size~\cite{piveteau2024}, so cutting is practical only when the seam, i.e., the bipartition between sub-circuits, carries bounded entanglement.
Bounded seam entanglement is, however, precisely the structure that implies classical simulability~\cite{cerezo2025}.
This raises the question of whether any circuit family is simultaneously cuttable at constant overhead and classically hard.
 
We establish three structural results.
Proposition~1 shows that matrix product state (MPS) and tree tensor network (TTN) circuits with constant seam bond dimension are simultaneously cuttable and classically simulable.
Proposition~2 shows this failure is not inherent: a two-block circuit with fixed seam width and increasing intra-block depth remains cheaply cuttable while its global simulation cost grows super-polynomially in $n$, confirmed by MPS simulation up to $n=100$.
Proposition~3 shows that the three-way question is provably negative for the brickwork family: the depth regimes required for MPS-hardness and for trainability are mutually exclusive.
Section~VI identifies the structural conditions for a positive answer and presents partial evidence via Clifford+T circuits.
Recent work has demonstrated that cutting-based distributed quantum learning models~\cite{marchisio2025,distest2026} and variational algorithms~\cite{bechtold2023qaoa,gentinetta2024} are operationally viable, has incorporated cut-count into circuit architecture search~\cite{qnas2026}, and has developed adaptive methods that discover low-entanglement cuts directly from a circuit's entanglement structure~\cite{ack2026}, but none addresses whether the resulting circuits possess quantum advantage, which Propositions 1 and 2 are the first to answer structurally.

\vspace{-0.45cm}
\section{Background}
\label{sec:bg}

Circuit cutting works by replacing each nonlocal gate~$\mathcal{U}$ at the
seam with a signed mixture of local operations,
$\mathcal{U} = \sum_i a_i \mathcal{F}_i$, where each $\mathcal{F}_i$ acts
entirely within one sub-circuit~\cite{piveteau2024}.
The expectation value of any observable is then reconstructed by Monte Carlo
sampling; the sampling overhead to achieve additive error~$\varepsilon$ is
$\#\text{shots} = \mathcal{O}\!\left(\gamma^{2k}/\varepsilon^2\right)$,
where $k$ is the number of cuts and $\gamma$ the per-cut $\ell_1$ factor,
with $\gamma\!=\!3$ for a CNOT gate cut and $\gamma\!=\!4$ for a wire cut
under local operations alone~\cite{piveteau2024}.
Classical communication between sub-circuits reduces wire-cut overhead
from $\mathcal{O}(16^k)$ to $\mathcal{O}(4^k)$, approaching the
teleportation limit~\cite{piveteau2024}; non-maximally entangled auxiliary
states interpolate smoothly between these regimes~\cite{bechtold2023}.
The cost grows \emph{exponentially in~$k$} but is otherwise
independent of circuit size: for fixed~$k$, the overhead is $O(1)$.

The minimum number of wire cuts at a bipartition equals
$\lceil\log_2\chi_{\rm seam}\rceil$, where $\chi_{\rm seam}$ is the Schmidt
rank of the state at the cut.
Cutting is therefore cheap if and only if the seam carries a bounded number of
ebits, i.e., $\chi_{\rm seam} = O(1)$.
MPS- and TTN-style ans\"{a}tze satisfy this by design: every inter-block bond
has bounded rank, so every bipartition aligned with a bond has
$\chi_{\rm seam} = O(1)$~\cite{pira2023,cong2019}.
For circuits where $\chi_{\rm seam}$ grows with system size, cutting overhead
scales accordingly and is no longer $O(1)$.
The generic factors $\gamma=3,4$ are worst-case upper bounds; concurrent
work~\cite{necaise2026} shows that an efficiently computable channel-level lower
bound $\gamma_{\rm Choi}$ follows from the operator-Schmidt (Choi) spectrum of
the cut unitary, generalizing the KAK cost to Gaussian fermionic unitaries.

In that vein, the same bounded entanglement that makes circuits cuttable also makes them
classically simulable.
For MPS and TTN circuits with bond dimension~$\chi$, standard sequential
tensor-network contraction evaluates any local expectation value in
$O(n\chi^3)$ time~\cite{verstraete2008}.
More broadly, Ref.~\cite{cerezo2025} shows that every standard
barren-plateau-free architecture lives in a polynomially sized subspace of
operator space, making its loss function efficiently computable either
classically (CSIM) or after a polynomial quantum data-acquisition step (QESIM).
The Lie-algebraic framework of~\cite{ragone2024} gives exact loss-variance
formulas for deep parameterized quantum circuits and unifies expressivity, locality, entanglement, and
noise as sources of barren plateaus. Bounded cuttability and classical
simulability share a common root: a small Schmidt rank at the seam limits both
the information communicated during cutting and the number of Schmidt values
needed to represent the state classically, biasing cut-compatible circuits
toward classical simulability.

The literature asks whether non-classically-simulable circuits can
simultaneously be useful for learning and cheaply cuttable~\cite{robustness2025};
adding trainability, with absence of barren plateaus as the standard proxy,
makes this an explicitly open three-way question. The analysis above suggests
the answer is negative for all standard cuttable families; the propositions
below make this precise and identify the structural condition for an
affirmative answer.

\vspace{-0.2cm}
\section{Negative Baseline}
\label{sec:neg}

MPS- and TTN-style circuits are natural candidates for cutting-based learning:
each bond is a seam with bounded Schmidt rank, so cutting is cheap at any depth.
One might expect them to be good vehicles for quantum advantage; the following
proposition shows the opposite, as the same global bounded-rank property also
enables efficient classical simulation.

\begin{proposition}[Negative baseline]\label{prop:neg}
  Let $C$ be an $n$-qubit MPS or TTN variational circuit with uniform bond
  dimension $\chi = O(1)$ (i.e., the Schmidt rank at every
  \emph{contiguous} bipartition is at most~$\chi$). Then:
  \begin{enumerate}[label=(\roman*),leftmargin=1.5em,topsep=1pt,itemsep=0pt]
    \item The circuit can be cut at any bipartition with sampling overhead
      $\mathcal{O}(\chi^{2\log_2\gamma}/\varepsilon^2)
       = \mathcal{O}(1/\varepsilon^2)$,
      independent of~$n$~\emph{\cite{piveteau2024}}.
    \item Expectation values of local observables are computable classically
      in $O(n\chi^3)$ time, placing $C \in \mathrm{CSIM}$
      ~\emph{\cite{verstraete2008,cerezo2025}}.
  \end{enumerate}
\end{proposition}

\begin{proof}
(i)~The Schmidt rank at any contiguous bipartition is at most~$\chi$ by the MPS/TTN
structure, so $m = \lceil\log_2\chi\rceil = O(1)$ wire cuts suffice.
Each cut multiplies the sampling overhead by~$\gamma^2$, giving total
overhead $\gamma^{2m}/\varepsilon^2 = O(1/\varepsilon^2)$, independent
of~$n$ and circuit depth.
(ii)~An MPS (or TTN, via tree contraction) with bond dimension~$\chi$
has at most~$\chi$ Schmidt values at every bipartition. Evaluating
$\langle O\rangle$ for any local observable costs $O(n\chi^3)$ operations
by sequential contraction~\cite{verstraete2008}, giving $O(n)$ for
constant~$\chi$ and placing~$C$ in CSIM.
\end{proof}

\begin{remark}
  At depth $d = O(\log n)$ with area-law initial states and local
  observables, these circuits are additionally free of barren plateaus
  \cite{cerezo2025}: the loss variance is $\Omega(1/\mathrm{poly}(n))$,
  which is equivalent to the canonical gradient-variance and
  cost-concentration criteria~\cite{arrasmith2022}.
  At shallow depth the architecture is therefore simultaneously cuttable at $O(1/\varepsilon^2)$ overhead,
  efficiently classically simulable, and free of barren plateaus.
\end{remark}

The key assumption in Proposition~\ref{prop:neg} is that $\chi=O(1)$ holds globally, across every bipartition.
Bounding the bond dimension only at the seam while allowing high internal entanglement can preserve cheap cuttability without efficient classical simulability.
We exploit this geometric insight to construct a cheaply cuttable circuit family with super-polynomial global MPS bond dimension.

\section{The Two-Block Wedge}
\label{sec:wedge}

We realize the separation between bounded seam entanglement and growing intra-block entanglement using a two-block circuit in which we can independently control the amount of seam gates across the $A$--$B$ boundary, and the intra-block gates acting entirely within each block, as formalized in Definition~\ref{def:2block}.

\begin{definition}[2-Block Circuit]\label{def:2block}
  $C(n,d,k)$ is an $n$-qubit circuit with:
  \begin{enumerate}
    \item \emph{Block~$A$}: qubits $1,\ldots,n/2$, a depth-$d$ 1D brickwork with
  nearest-neighbour two-qubit gates;
    \item \emph{Block~$B$}: qubits $n/2{+}1,\ldots,n$, same structure;
    \item \emph{Seam}: $k$ two-qubit gates on the $A$--$B$ boundary, applied
  after all block-$A$ gates and before all block-$B$ gates.
  The seam width~$k$ is fixed; it does not scale with either~$n$ or~$d$.
  \end{enumerate}
\end{definition}

\begin{proposition}[Two-block wedge]\label{prop:wedge}
  Let $C(n,d,k)$ be as in Definition~\ref{def:2block} with generic (hence non-Clifford)
  two-qubit intra-block gates and fixed~$k$. Then:
  \begin{enumerate}[label=(\roman*),leftmargin=1.5em,topsep=1pt,itemsep=0pt]
    \item The circuit-cutting cost at the seam is
      $\mathcal{O}(\gamma^{2k}/\varepsilon^2)$, independent of~$d$
      and~$n$~\emph{\cite{piveteau2024}}.
    \item With high probability over the intra-block gate choices, the
      entanglement entropy at any fixed internal bipartition of each block
      satisfies $S_{\rm block}(d) = \Theta(\min(d,\,n/4))$~\emph{\cite{nahum2017}}.
      Any MPS approximation to the global state at fixed precision requires
      $\chi_{\rm global} = \Omega(2^{\Theta(d)})$~\emph{\cite{verstraete2008,eisert2010}}.
  \end{enumerate}
  For $d = \omega(\log n)$: $\chi_{\rm global} = n^{\omega(1)}$, so global MPS
  contraction costs $\omega(\mathrm{poly}(n))$, while the seam cutting cost
  remains $O(1/\varepsilon^2)$, independent of~$n$.
  The two-block wedge is non-empty.
\end{proposition}

\begin{proof}
For~(i), prior to the seam gates the state is a product across the
$A$--$B$ bipartition with Schmidt rank~1. Each seam gate raises this rank
by at most one bit, giving $\chi_{\rm seam} \leq 2^k$. After the seam,
block-$B$ gates act unitarily within~$B$; they transform Schmidt vectors on
the $B$-side but leave Schmidt values unchanged, so $\chi_{\rm seam}$ is
bounded at $O(1)$ for all~$d$. Thus $k$ wire cuts suffice, with overhead
$\gamma^{2k}/\varepsilon^2 = O(1/\varepsilon^2)$.

For~(ii), consider the bipartition $\{1,\ldots,n/4\}$ vs.\ the rest, a cut
strictly inside block~$A$. The seam gates act on qubit~$n/2$ and adjacent
pairs; block-$B$ gates act entirely within~$B$. Neither touches this internal
bipartition, whose entanglement is therefore determined solely by the block-$A$
brickwork. By~\cite{nahum2017}, the middle bipartition entropy satisfies
$S_{\rm block}(d) = \Theta(d)$ with high probability for generic SU(4) brickwork
with $d \leq n/4$. For $d = O(\log n)$ this gives
$2^{O(\log n)} = \mathrm{poly}(n)$ bond dimension and polynomial contraction
cost. For $d = \omega(\log n)$, any $\varepsilon$-accurate global MPS requires
$\chi_{\rm global} \geq 2^{\Theta(d)} = n^{\omega(1)}$
\cite{verstraete2008,eisert2010}, so contraction costs $\omega(\mathrm{poly}(n))$.
The seam and the internal-$A$ bipartition are geometrically disjoint, so the
two conclusions are independent.
\end{proof}

\begin{remark}[Blockwise simulation does not help]
The super-polynomial bond dimension of Proposition~\ref{prop:wedge}(ii) sits at
an internal bipartition strictly inside block~$A$, not at the seam. Splitting at
the $A$--$B$ seam and simulating each half separately does not avoid it:
block~$A$ alone still requires $\chi = 2^{\Theta(d)}$. The separation is against
global sequential MPS contraction, not circumvented by seam-blockwise
decomposition.
\end{remark}

Proposition~\ref{prop:wedge} makes depth~$d$ an independent lever for classical
hardness at fixed cutting cost; Proposition~\ref{prop:notriple} asks whether
that depth still preserves a navigable loss landscape.

\begin{proposition}[Three-way impossibility]\label{prop:notriple}
  For $C(n,d,k)$ as in Definition~\ref{def:2block} with area-law initial
  state~$\rho$ and local observable~$O$, the three conditions
  \begin{enumerate}[label=(\roman*),leftmargin=1.5em,topsep=1pt,itemsep=0pt]
    \item cheaply cuttable: overhead $O(\gamma^{2k}/\varepsilon^2)$,
    \item not efficiently simulable via global MPS, and
    \item trainable: loss variance $\Omega(1/\mathrm{poly}(n))$,
  \end{enumerate}
  cannot all hold simultaneously. Condition~(i) holds for all~$d$
  by Proposition~\ref{prop:wedge}(i). Conditions~(ii) and~(iii)
  require $d = \omega(\log n)$ and $d = O(\log n)$ respectively;
  no depth satisfies both.
\end{proposition}

\begin{proof}
Condition~(i) holds unconditionally by Proposition~\ref{prop:wedge}(i).
Condition~(ii) requires $d = \omega(\log n)$ by Proposition~\ref{prop:wedge}(ii).
Condition~(iii) requires $d = O(\log n)$: Ref.~\cite{cerezo2025}
proves that shallow 1D brickwork with area-law~$\rho$ and local~$O$ satisfies
loss variance $\Omega(1/\mathrm{poly}(n))$---equivalent to absence of
exponentially vanishing gradients~\cite{arrasmith2022}---while at
$d = \omega(\log n)$ the same entanglement growth that drives hardness also
causes exponential gradient concentration~\cite{cerezo2025,arrasmith2022}. The threshold $d^* = \Theta(\log n)$ is where
both transitions occur simultaneously; no depth satisfies both (ii)
and~(iii).
\end{proof}
Condition~(iii) characterises trainability via gradient concentration at random
initialisation; it rules out gradient-based training in the standard operational
model, but does not preclude structured initialisation or non-gradient optimisers.

\vspace{-0.1cm}
\section{Numerical Evidence}
\label{sec:num}

\begin{figure}[t]
  \centering
  \includegraphics[width=\columnwidth]{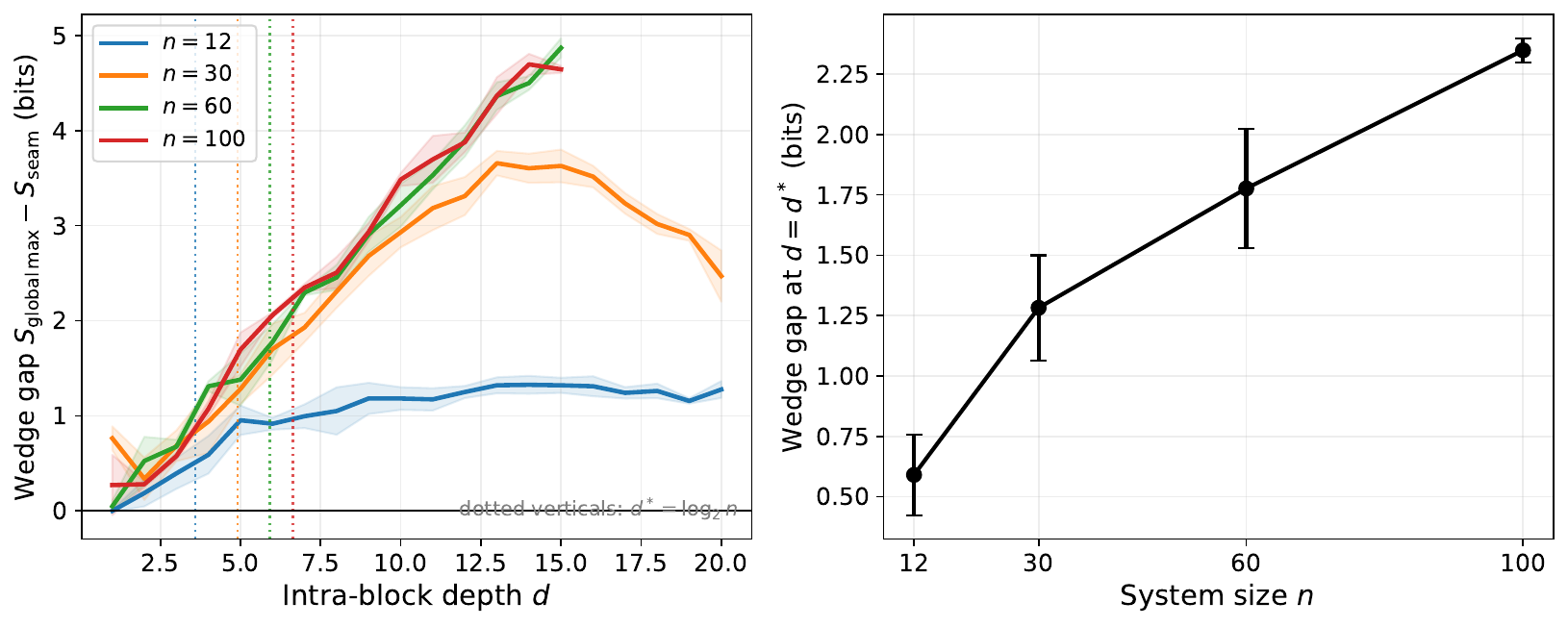}
  \caption{%
    MPS verification of the wedge across $n \in \{12, 30, 60, 100\}$
    ($k=1$ seam CNOT, random SU(4) brickwork, bond dimension $\chi\leq 256$).
    \emph{Left}: wedge gap $S_{\rm global\,max} - S_{\rm seam}$ vs.\
    intra-block depth~$d$; dotted verticals mark $d^* = \log_2 n$ per curve.
    $S_{\rm seam}$ stays below~$1$ bit at every~$n$ and~$d$; the gap opens
    and grows with~$n$.
    \emph{Right}: gap at $d = d^*$ vs.\ $n$, growing as $\Theta(\log_2 n)$.
    The wedge is a structural property, not a finite-size artefact.
  }
  \label{fig:wedge}
  \vspace{-0.6cm}
\end{figure}

Proposition~\ref{prop:wedge} predicts $S_{\rm seam} \leq k$ bits for all $d$
and $n$, while $S_{\rm global\,max}$ grows with $d$ inside each block.
To verify this at operationally relevant scales we simulate $C(n, d, 1)$ via
MPS (quimb~\cite{gray2018}), sweeping $n \in \{12, 30, 60, 100\}$ and
$d = 1, \ldots, d_{\rm max}$, with bond dimension truncated at
$\chi \leq 256$; this gives a lower bound on $S_{\rm global\,max}$ while
leaving $S_{\rm seam}$ unaffected (its true rank is at most $2^k = 2$).

Figure~\ref{fig:wedge} confirms both predictions at all scales.
$S_{\rm seam}$ saturates near $1$ bit at each~$n$, independent of depth,
while $S_{\rm global\,max}$ grows continuously with~$d$~\cite{nahum2017}.
At the trainability threshold $d^* = \log_2 n$, the measured gaps are
$0.59$, $1.28$, $1.78$, and $2.35$ bits for $n = 12, 30, 60, 100$
respectively — growing as $\Theta(\log_2 n)$, consistent with
$S_{\rm global\,max}(d^*) \approx d^* = \log_2 n$ and $S_{\rm seam} \approx 1$.
The right panel confirms this gap is a structural consequence of the geometric
decoupling of Proposition~\ref{prop:wedge}, not a finite-size artefact.

\vspace{-0.1cm}
\section{Stabiliser Hardness and Open Questions}
\label{sec:open}

Proposition~\ref{prop:notriple} rules out all three conditions when hardness
means MPS-hardness, a resource tied to depth.
Clifford$+T$ circuits at shallow depth carry low entanglement but accumulate
\emph{magic} as the $T$-count grows.
Two regimes must be distinguished here.
In the \emph{fault-tolerant} setting, $T$-gates are expensive: they require
magic state distillation, a substantial overhead absent from Clifford gates.
In the \emph{NISQ} setting, $T$-gates are native hardware operations with no
additional cost over any other single-qubit rotation, so the quantum execution
overhead is zero.
The $\Omega(3^t)$ lower bound below is a \emph{classical-simulation} complexity
statement (it measures the cost of classically simulating the circuit) and applies in
both regimes; the NISQ/fault-tolerant distinction affects only the cost of
\emph{running} the circuit on quantum hardware.
Because stabiliser simulation cost scales with $T$-count rather than depth,
the depth conflict no longer applies.

\begin{corollary}[Stabiliser-hardness instance]\label{cor:cliffT}
  Let $C(n,d,k)$ be as in Definition~\ref{def:2block} with Clifford$+T$
  intra-block gates at fixed depth $d = O(\log n)$, area-law initial
  state~$\rho$, and local observable~$O$. As the $T$-count $t$ per block
  grows with~$n$, conditions~(i) and~(iii) of Proposition~\ref{prop:notriple}
  hold simultaneously, and the stabiliser-simulation cost of the global state
  grows at least as $\Omega(3^t)$.
\end{corollary}

\begin{proof}
  Condition~(i) holds for all~$d$ by Proposition~\ref{prop:wedge}(i).
  Condition~(iii) holds because $d = O(\log n)$: the argument of
  Proposition~\ref{prop:notriple} applies to Clifford$+T$ circuits since the
  trainability guarantee of~\cite{cerezo2025} depends on circuit depth and
  observable locality, not on the specific gate set.
  For the hardness claim, Ref.~\cite{howard2017} proves that
  the stabiliser $\ell_1$-norm of $t$ independent $T$-gate magic states is
  \emph{exactly} $(\sqrt{3})^t = 3^{t/2}$, and that this quantity is a
  lower bound on the $\ell_1$ norm of any quasi-probability stabiliser
  decomposition of those states. Any stabiliser simulation via
  quasi-probability sampling therefore requires at least
  $(3^{t/2})^2 = \Omega(3^t)$ samples.
\end{proof}

\vspace{-0.2cm}

Figure~\ref{fig:corollary1} verifies Corollary~\ref{cor:cliffT} on
$C(12,d{=}3,k{=}1)$ (50 seeds, 200 parameter samples each,
$t \in \{0,2,\ldots,18\}$). Across all~$t$, $S_{\rm seam}\in[0.83,0.90]$ bits
(below the $k{=}1$ ceiling) and $\mathrm{Var}[\langle Z_5\rangle]\in[0.13,0.21]$,
confirming conditions~(i) and~(iii) are independent of $T$-count, while the
stabiliser sampling overhead $3^t$ grows from $1$ to ${\approx}3.9\times10^8$.
\begin{figure}[t]
  \centering
  \includegraphics[width=\columnwidth]{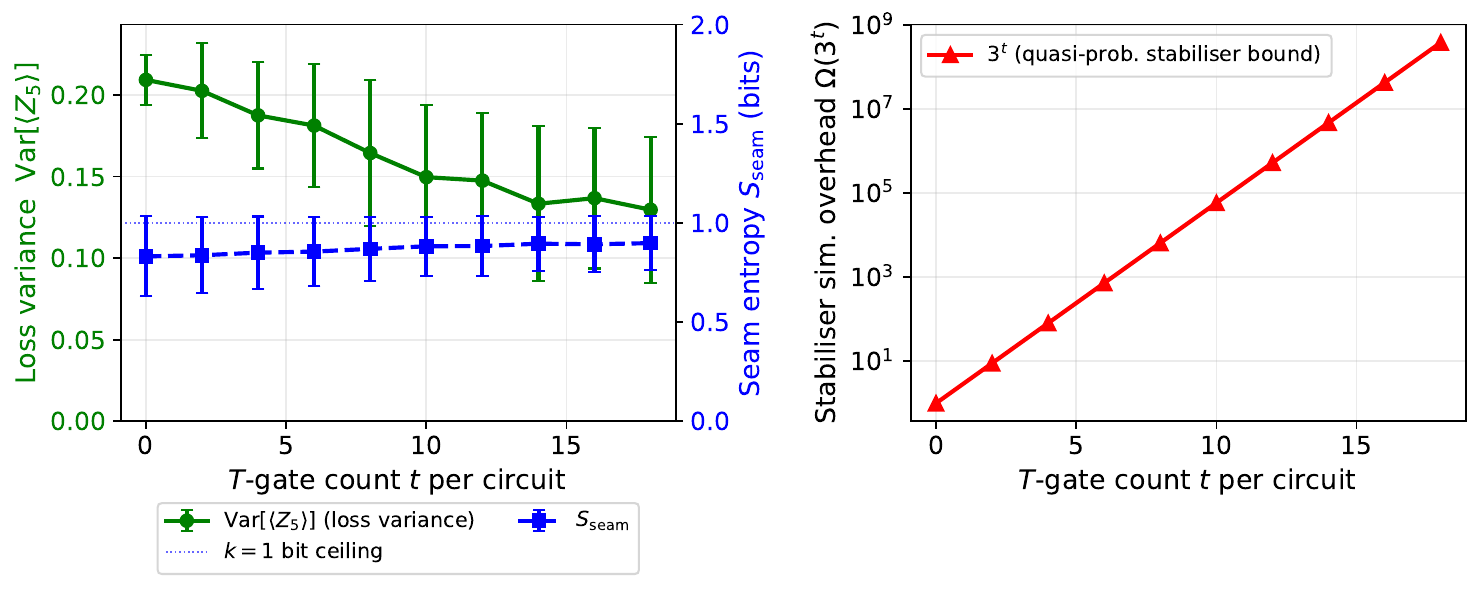}
  \caption{%
    Corollary~\ref{cor:cliffT} on $C(12,d{=}3,k{=}1)$;
    50 seeds $\times$ 200 samples.
    \emph{Left}: $S_{\rm seam}$ (blue, right axis) stays below the $k{=}1$ ceiling
    and loss variance (green, left axis) stays in $[0.13,0.21]$ across all $T$-counts.
    \emph{Right}: sampling overhead $3^t$ rises eight orders of magnitude while
    both metrics stay flat, confirming magic as the independent hardness lever.
    Mean $\pm$ 1\,s.d.}
  \label{fig:corollary1}
  \vspace{-0.6cm}
\end{figure}

Note that simulability by
\emph{one} classical method already implies trainability (caveat of Ref.~\cite{cerezo2025}), so a trainable circuit
need not be hard for \emph{all}. Shallow Clifford$+T$ circuits
at $d = O(\log n)$ keep $\chi_{\rm global} = \mathrm{poly}(n)$, remaining
efficiently MPS-simulable, hence trainable and stabiliser-hard. 
Universal classical hardness at shallow depth remains open.

New trainability arguments that survive the deep regime of
Proposition~\ref{prop:notriple} would yield a positive answer.
Note too that $t$ and $d$ are not fully independent: a depth-$d$ circuit on
$n$ qubits holds at most $O(dn)$ $T$-gates, so Corollary~\ref{cor:cliffT} is
best read in the regime $t = o(dn)$, where $d = O(\log n)$ is maintained; beyond
it, added $T$-gates force $d$ to grow and re-enter the barren-plateau regime.
It remains open whether circuit-cutting-based quantum advantage requires seam--interior decoupling.

\vspace{-0.1cm}
\section{Conclusion}
\label{sec:conclusion}

This paper asks whether a variational circuit can be cheaply cuttable,
classically hard, and trainable at the same time. The answer is set by
entanglement geometry.
Seam gates and intra-block gates act on disjoint parts of the circuit,
making them independent knobs.
MPS and TTN architectures lose this independence by imposing a global
bond-dimension bound, making them simultaneously cuttable, simulable, and
trainable by Proposition~\ref{prop:neg}.
The two-block family of Proposition~\ref{prop:wedge} restores it and exposes
a two-tier hardness picture.
When entanglement is the hardness resource, depth is the lever, but the depth
that pushes interior entanglement past the MPS threshold also induces barren
plateaus, so by Proposition~\ref{prop:notriple} hardness and trainability cannot
coexist.
Switching to magic breaks this coupling.
Shallow Clifford+T circuits stay MPS-simulable and trainable while their
stabiliser-simulation cost grows as $\Omega(3^t)$ in the $T$-count, independent
of depth, so Corollary~\ref{cor:cliffT} sidesteps Proposition~\ref{prop:notriple}
by changing the resource.
Magic supplies a hardness lever free of the depth conflict but not of cost,
whether through magic-state distillation (fault-tolerant) or the still-open
resource-theoretic status of stabiliser simulation (NISQ).

The code used in this work can be found at {\hypersetup{urlcolor=[rgb]{0,0.5,0}}\url{https://github.com/grageragarces/Entanglement-geometry}}

\vspace{-0.15cm}
\section*{Acknowledgments} 
MGG is funded by the EPSRC UK Quantum Technologies Programme under grant EP/T001062/1
and VeriQloud. 
SD is funded by the ETH Zurich Quantum Center.
LP is supported by the National Research Foundation, Singapore through the National Quantum Office, hosted in A*STAR, under its Centre for Quantum Technologies Funding Initiative (S24Q2d0009). LP is also partially supported by A*STAR under its YIRG M25N8c0131.

\vspace{-0.15cm}
\bibliographystyle{IEEEtran}
\bibliography{refs}

\end{document}